\documentclass[11pt,lettersize]{article}

\usepackage{authblk}
\usepackage{fullpage}
\usepackage{amsmath,amsthm,amssymb}
\usepackage{todonotes}
    \presetkeys{todonotes}{inline}{} 
\usepackage{complexity} 
\usepackage{optidef}
\usepackage{thm-restate}
\usepackage{hyperref}
\usepackage{cleveref} 

\DeclareMathOperator*{\argmax}{arg\,max}

\usepackage[noend]{algpseudocode}
\usepackage{algorithm}

\usepackage{tikz}
\usetikzlibrary{calc,fit,arrows.meta, positioning}
\usetikzlibrary{arrows.meta, positioning, shapes}
\usetikzlibrary{decorations.pathreplacing} 
\usetikzlibrary{patterns}

\newtheorem{theorem}{Theorem}[section]
\newtheorem{proposition}[theorem]{Proposition}

\newtheorem{lemma}[theorem]{Lemma}

\theoremstyle{definition}
\newtheorem{definition}[theorem]{Definition}
\newtheorem{example}[theorem]{Example}

\usepackage{todonotes}
    \presetkeys{todonotes}{inline}{}

\newcommand{\vf}{v} 
\newcommand{\negval}{\vf^-}
\newcommand{\SW}{\mathcal{SW}}
\newcommand{\CW}{\mathcal{CW}} 
\newcommand{\TV}{\mathcal{V}} 

\newcommand{\ER}[2]{G =(#1,#2)} 
\newcommand{\twoparts}{\Pi_N^{(2)}} 
\newcommand{\approxwelfare}{\textsc{ApproxWelfare}} 
\newcommand{\MAXCL}{\textsc{MaximumClique}}
\newcommand{\eps}{\varepsilon}

\newcommand{\EV}{\mathbb E}
\newcommand{\orderof}{\mathcal O} 

\usepackage{xcolor}
\definecolor{textBlue}{RGB}{0, 0, 255} 
\hypersetup{colorlinks,linkcolor=textBlue,citecolor=textBlue,urlcolor=textBlue}

\title{Welfare Approximation in Additively Separable Hedonic Games}

\author[1]{Martin Bullinger}
\author[2]{Vaggos Chatziafratis}
\author[3]{Parnian Shahkar}
\affil[1]{University of Oxford}
\affil[2]{University of California, Santa Cruz}
\affil[3]{University of California, Irvine}

\date{}

\begin{document}

\maketitle
\begin{abstract}
Partitioning a set of $n$ items or agents while maximizing the value of the partition is a fundamental algorithmic task.
We study this problem in the specific setting of maximizing social welfare in additively separable hedonic games.
Unfortunately, this task faces strong computational boundaries: 
Extending previous results, we show that approximating welfare by a factor of $n^{1-\eps}$ is \NP-hard, even for severely restricted weights.
However, we can obtain a randomized $\orderof(\log n)$-approximation on instances for which the sum of input valuations is nonnegative.
Finally, we study two stochastic models of aversion-to-enemies games, where the weights are derived from Erd\H{o}s-R\'{e}nyi or multipartite graphs. We obtain constant-factor and logarithmic-factor approximations with high probability.
\end{abstract}

\section{Introduction}

Partitioning a set of items or agents, say humans or machines, is a fundamental problem that has been studied across many disciplines such as computer science, economics, or mathematics.
For instance, it is relevant in the context of clustering, an important task in machine learning with far-reaching applications like image segmentation \cite{DMC15a}, or for community detection, which helps in understanding networks, e.g., of societies or physical systems \cite{Newm04a}.

Our paper takes a game-theoretic perspective and considers the prominent model of \emph{additively separable hedonic games} \cite{BoJa02a}.
We assume that there is a set of agents that has to be partitioned into coalitions and agents have preferences over the coalitions that they are part of \cite{DrGr80a}.
Preferences are given by a weighted graph, where the agents are the vertices and the edge weights encode the valuation between agents.
The utility of an agent for a coalition is the sum of weights of edges towards members of this coalition.
This class of games is quite expressive and contains more structured subclasses of games.
For instance, an agent might divide the other agents into friends and enemies and could simply try to maximize the number of friends within their coalition while minimizing the number of enemies.
A priority between these two objectives can be captured by the exact edge weights: for example, if there is a large negative weight for enemies and a small positive weight for friends, then minimizing enemies is much more important than maximizing friends, as conceptualized in so-called \emph{aversion-to-enemies games}~\cite{DBHS06a}.

A fundamental quantity for evaluating a possible output is its \emph{social welfare} (also called utilitarian welfare) which is the sum of all agents utilities.
Unfortunately, maximizing this quantity faces significant computational boundaries.
Aziz et al. \cite{ABS11c} show that it is \NP-hard to maximize, and, even worse, approximating maximum welfare by a factor of at least $n^{1-\eps}$ is \NP-hard for any $\eps > 0$ \cite{FKV22a}.
Our paper aims at circumventing this computational boundary.

First, we investigate the inapproximability of maximum welfare.
Notably, the result of Flammini et al. \cite{FKV22a} is for aversion-to-enemies games, which use valuations $-n$ and $1$, i.e., the negative valuation is dependent on the number of agents~$n$.
We complement this by showing an $n^{1-\eps}$-inapproximability result on instances in which the valuations are restricted to $\{-\negval, 0, 1\}$, where $\negval \ge 1$ is an arbitrary but fixed (and, therefore, globally bounded) number.\footnote{By rescaling valuations, this is equivalent to assuming that, in addition to a neutral valuation of~$0$, there is a single positive and negative valuation, where the former is bounded by the absolute value of the latter.}
This sounds discouraging but it strengthens the impression that negative valuations seem to be the reason for computational boundaries.

In the remainder of the paper, we provide several possibilities to achieve better approximation guarantees.
First, we consider the restricted domain of games in which the sum of all valuations is nonnegative.
This assumption still allows for the existence of rather negative valuations, however, it disallows an overall bias towards negative valuations. 
We make use of a result from the correlation clustering literature \cite{ChaWi04} to prove the existence of a randomized algorithm that approximates social welfare by a factor of $\orderof(\log n)$.

Second, we consider two stochastic models of aversion-to-enemies games in which we achieve approximation guarantees with high probability.
We start by assuming a basic model where valuations originate from an Erd\H{o}s-R\'{e}nyi graph.
We show that a constant approximation of maximum welfare is possible.
Subsequently, we define a stochastic model inspired by team management where every agent has a role, such as project manager, software engineer, UX designer, or marketing specialist.
Coalitions represent teams and each role should be present in a team at most once.
This scenario can be conceptualized by making agents with the same role mutually incompatible by introducing large negative valuations.
In other words, the compatibility of agents is captured by a multipartite graph where the roles induce a partition of the vertices.
However, in reality, even agents of different roles might be incompatible for various reasons.
We model this by introducing a parameter $p$ that captures the probability of agents being incompatible.
In our stochastic model, every pair of agents admitting different roles are incompatible independently.
Based on the magnitude of $p$ we obtain perturbation regimes that lead to different approximation guarantees.
In the low perturbation regime, we can approximate maximum welfare by a constant factor, whereas a high perturbation regime allows for a $\log n$-approximation.

\section{Related work}

Hedonic games were introduced by Dr{\`e}ze and Greenberg \cite{DrGr80a} as an ordinal model of coalition formation, in which agents state their preferences as rankings over coalitions.
Their broad consideration started, however, only $20$ years later \cite{CeRo01a,BKS01a,BoJa02a}.
Much of their popularity today is due to the introduction of additively separable hedonic games by Bogomolnaia and Jackson \cite{BoJa02a} in this era.
An introduction to hedonic games is provided in the book chapters by Aziz and Savani \cite{AzSa15a} and Bullinger et al. \cite{BER24a}.

While these first papers were in the realm of economic theory, they soon sparked a broader consideration of hedonic games in computer science.
This led to increased attention of algorithmic properties of solution concepts, including their computational complexity \cite{CeHa02a,Ball04a}.
Social welfare was first realized to be a demanding objective by Aziz et al. \cite{ABS11c} who showed that it is \NP-hard to compute even if valuations are restricted to be only $-1$ or $1$.
Subsequently, Flammini et al. \cite{FKV22a} significantly strengthened this to the $n^{1-\eps}$-inapproximability result for aversion-to-enemies games mentioned in the introduction. 

Beyond social welfare, other welfare objectives have been explored.
Some early papers on hedonic games already studied Pareto optimality, a less demanding notion of welfare studied throughout economics \cite{DrGr80a,BoJa02a}. 
Pareto-optimal coalition structures can be computed in polynomial time under fairly general assumptions including symmetric valuations \cite{ABS11c,Bull19a}.
However, this yields no approximation of social welfare because Pareto-optimal outcomes may have negative social welfare~\cite{EFF20a}.
Moreover, Aziz et al. \cite{ABS11c} also considered egalitarian welfare, which aims at maximizing the utility of the worst-off agent.

Despite its challenges for the offline model, welfare approximation has also been studied in an online variant of additively separable hedonic games \cite{FMM+21a,BuRo23a}.
Flammini et al. \cite{FMM+21a} consider a general model where no finite competitive ratio is possible if the utility range is unbounded.
Moreover, Bullinger and Romen \cite{BuRo23a} study a model where the algorithm is allowed to dissolve coalitions into singleton coalitions, which allows to achieve a coalition structure with a social welfare that is at most a factor of $\Theta(n)$ worse than the maximum possible welfare.
In particular, they show that maximum weight matchings achieve an $n$-approximation of social welfare \cite{BuRo23a}.
This essentially matches the aforementioned inapproximability by a factor of $n^{1-\eps}$ \cite{FKV22a}.
Finally, social welfare has been considered in a mechanism design perspective aiming at strategyproof preference elicitation \cite{FKMZ21a,FKV22a}.

Beyond welfare, the most common objectives in hedonic games are notions of stability \cite{BoJa02a,SuDi10a,ABS11c,Woeg13a,GaSa19a,BBT23a,BBW21b,BrBu20a}.
Rather than the global guarantees provided by welfare notions, stability assumes a more strategic perspective in that it requires the absence of beneficial deviations by single agents or groups of agents.
Single-deviation stability often leads to \NP-completeness \cite{SuDi10a,BBT23a}, whereas group stability can even be $\Sigma_2^p$-complete \cite{Woeg13a}.
Interestingly, symmetric valuations lead to the existence of stable outcomes based on single-agent deviations \cite{BoJa02a}, but their computation is still infeasible.
It is \PLS-complete, i.e., complete for the complexity class capturing problems that guarantee solutions based on local search algorithms \cite{GaSa19a}.
An interesting objective that combines ideas of stability and global guarantees is popularity \cite{ABS11c,BrBu20a}, which is akin to weak Condorcet winners as 
studied in social choice theory \cite{BCE+14a}.

While all the literature discussed so far considers a deterministic model, stochastic models have been studied to some extent \cite{FFKV23a,BuKr24a}.
In particular, Bullinger and Kraiczy \cite{BuKr24a} show how to obtain stable outcomes if valuations are drawn uniformly at random.
Their algorithm runs in three stages, the first of which will turn out to be useful in obtaining welfare guarantees as well, see \Cref{sec:ERgraphs}.
By contrast, Fioravanti et al. \cite{FFKV23a} consider a deterministic game model and aim at computing outcomes that are stable with high probability.

Finally, hedonic games are also related to other graph partitioning problems such as correlation clustering.
The input typically consists of a complete graph with edges labeled as ``$+$'' or ``$-$'' to indicate similarity or dissimilarity, respectively~\cite{BBC04a,Swam04a,DEFI06a}. 
The goal is to find a partition that maximizes agreements as measured by the sum of ``$+$'' edges inside clusters plus ``$-$'' edges across different clusters. 
Other objectives where the goal is to minimize errors of the partition, measured by ``$-$'' edges within clusters plus ``$+$'' edges across clusters have also been extensively studied~\cite{charikar2005clustering}. 
By contrast, our social welfare objective in hedonic games is different in that it accounts only for the edges within the coalitions (in particular, it ignores the edges across different coalitions).

Going beyond worst-case analysis, the two stochastic models we study for hedonic games in the second part of our paper relate to various stochastic models with random (or semi-random) edges that have been proposed for correlation clustering. 
For example, Mathieu and Schudy \cite{MaSc10a} investigates a noisy model on complete graphs, where they start from an arbitrary partition of the vertices into clusters and for each pair of vertices, the edge information (either $1$ or $-1$) is corrupted independently with probability $p$. 
Other average-case models and extensions to arbitrary graphs (not necessarily complete) have been studied by Makarychev et al. \cite{makarychev2015correlation}, with the goal of designing provably good approximation algorithms.

\section{Preliminaries}

Consider a finite set \(N\) of \(n := |N|\) agents. 
A \emph{coalition} is a nonempty subset of $N$. 
We denote by $\mathcal{N}_i := \{S\subseteq N\colon i\in S\}$ the set of all coalitions that agent $i$ belongs to. 
A \emph{coalition structure} (or \emph{partition}) is a partition $\pi$ of $N$ into coalitions, i.e., $\bigcup_{C\in \pi}C = N$ and for each pair of coalitions $C,C'\in \pi$ with $C\neq C'$ it holds that $C \cap C' = \emptyset$. 
Note that there is no bound on the number or size of coalitions.
For an agent $i\in N$, we denote by $\pi(i)$ the coalition that $i$ belongs to in $\pi$. 
We denote the set of all partitions of $N$ by $\Pi_N$, and the set of all partitions containing exactly two coalitions as $\twoparts$, i.e., $\twoparts := \{\pi \in \Pi_N\colon |\pi| = 2\}$.
A coalition is called a \emph{singleton coalition} if it contains exactly one agent. 
The partition where every agent is in a singleton coalition is called the \emph{singleton partition}.

In a hedonic game, every agent possesses preferences over the coalitions in $\mathcal{N}_i$.
We use the model of additively separable hedonic games by Bogomolnaia and Jackson \cite{BoJa02a} in which these preferences are obtained from cardinal valuations that can be encoded by a complete and directed weighted graph.
Formally, a \emph{cardinal hedonic game} is a pair $(N,u)$ where $u = (u_i\colon \mathcal{N}_i\to\mathbb{R})_{i\in N}$ is a vector of \emph{utility functions}. 
An \emph{additively separable hedonic game} (ASHG) is specified by a vector $\vf = (\vf_i\colon N\to \mathbb R)_{i\in N}$ of (single-agent) \emph{valuation functions}. 
It is then defined as the cardinal hedonic game $(N,u)$, where for any agent $i\in N$ and coalition $C\in \mathcal N_i$, it holds that
$$u_i(C) := \sum_{j\in C}\vf_i(j)\text.$$

In words, the utility of a coalition is derived from single-agent values, which are aggregated by summing the values of the agents in this coalition.
Since valuation functions fully specify an ASHG, we also speak of the ASHG $(N,\vf)$.
Note that ASHGs can be encoded by a weighted graph where agents are vertices and edge weights are given by the valuations.

We extend utilities over coalitions to utilities over partitions by defining $u_i(\pi) := u_i(\pi(i))$.
Given coalitions $C,C'\in \mathcal{N}_i$, we say that agent $i$ \emph{prefers} $C$ over $C'$ if $u_i(C)\ge u_i(C')$.
Moreover, we say that $i$ \emph{strictly prefers} $C$ over $C'$ if $u_i(C) > u_i(C')$.
We use the same terminology for partitions.
Given an ASHG $(N,\vf)$ and a partition $\pi$, we define its \emph{social welfare} as
$$\SW(\pi) := \sum_{i\in N} u_i(\pi) = \sum_{C\in \pi\colon i,j\in C} \vf_i(j) \text.$$

Hence, the social welfare is the sum of the utilities which, in an ASHG, is equivalent to the sum of all valuations between agents in the same coalition.
We denote by $\pi^*$ a partition that maximizes $\SW$.

ASHGs admit various interesting subclasses when restricting valuations.
Following Dimitrov et al. \cite{DBHS06a}, an ASHG $(N,\vf)$ is called an \emph{aversion-to-enemies game} if $v_i(j)\in \{-n,1\}$ for all $i,j\in N$.
An ASHG $(N,\vf)$ is called \emph{symmetric} if for each pair of agents $i,j\in N$, it holds that $\vf_i(j) = \vf_j(i)$.
We write $\vf(i,j)$ for the symmetric valuation function between $i$ and $j$.
In this paper, we restrict attention to symmetric ASHGs.\footnote{\label{fn:sym}
In general ASHGs, symmetry is without loss of generality when reasoning about social welfare as the welfare remains the same if we replace $v_i(j)$ and $v_j(i)$ by $\frac{v_i(j) + v_j(i)}{2}$, see, e.g., \cite{Bull19a}.
However, this is not the case for aversion-to-enemies games as the symmetrization may leave this game class.}

Consider an ASHG $(N,\vf)$ and 
an approximation ratio $c\ge 1$.
A partition $\pi$ is said to provide a \emph{$c$-approximation to maximum welfare} if $c\cdot \SW(\pi) \ge \SW(\pi^*)$.
We are interested in the following computational problem.

\begin{table}[h!]
\centering
\renewcommand{\arraystretch}{1.5}
\begin{tabular}{|l|p{.7\columnwidth}|}
\hline
\multicolumn{2}{|c|}{\textbf{$c$-\approxwelfare}} \\ 
\hline
\textbf{Given:} & ASHG $(N, \vf)$. \\
\hline
\textbf{Task:} & Compute a partition with a $c$-approximation to maximum welfare.\\
\hline
\end{tabular}
\end{table}

We consider both deterministic and randomized algorithms and aim at efficient algorithms.
For $c\ge 1$, a polynomial-time algorithm is called a \emph{$c$-approximation algorithm} for maximizing welfare if it solves $c$-\approxwelfare.
For randomized algorithms, the expected running time has to be  bounded by a polynomial and the produced partition has to provide a $c$-approximation to maximum welfare in expectation. 
Note that we allow (and frequently assume) that the factor $c$ depends on~$n$. 

Finally, given an ASHG $(N,\vf)$, we define its \emph{total value} as
$$\TV(N,\vf) := \sum_{i,j\in N} \vf_i(j)\text.$$
We will obtain good approximation guarantees by restricting attention to ASHGs with nonnegative total value.

In this paper, we use $[k]$ to represent the set $\{1, \dots, k\}$.
Moreover, in asymptotic statements, we state logarithms without base. 
They can be assumed to have base $e$.

\section{Deterministic Games}
\label{sec:deterministic}
Recall that Aziz et al. \cite{ABS11c} show that maximizing social welfare is NP-hard. 
Their result even holds for symmetric valuations restricted to $\{-1,1\}$. 
Moreover, Flammini et al. \cite{FKV22a} prove that approximating social welfare by a factor of $n^{1-\eps}$ is \NP-hard for aversion-to-enemies games, i.e., when valuations are in the set $\{-n,1\}$.
In this section, we will significantly deepen the understanding of welfare approximability around this result.
First, we show that the result does not rely on unbounded negative weights by providing a reduction for valuations restricted to $\{-\negval,0,1\}$ where $\negval \ge 1$.
In particular, this means that unbounded negative weights or negative weights of absolute value much larger than the value of positive weights are not necessary.
Subsequently, we will show how to circumvent the inapproximability result for ASHG with nonnegative total value.

\subsection{Welfare Inapproximability for Restricted Valuations}

We now prove our inapproximability result.\footnote{We would like to thank Abheek Ghosh for the proof idea of \Cref{thm:hardnessapprox}.}

\begin{theorem}\label{thm:hardnessapprox}
Let $\eps > 0$ and $\negval \ge 1$. Then, unless \P{} $=$ \NP{}, $n^{1-\eps}$-\textsc{ApproxWelfare} cannot be solved in polynomial time for symmetric ASHGs with valuations in the set $\{-\negval, 0, 1\}$.
\end{theorem}

\begin{proof}
    Let $\eps > 0$ and $\negval \ge 1$.
    We reduce from the $n^{1-\eps}$-approximate {\MAXCL} problem.
    The input is an unweighted graph $G$ and the task is to compute a clique $C$ of $G$ with $n^{1-\eps}\cdot |C| \ge \mu^*$, where $\mu^*$ is the size of a maximum clique of $G$.
    Unless \P{} $=$ \NP{}, this problem cannot be solved in polynomial time \cite{Zuck06a}.

    We now describe the reduction.
    Assume that we are given an unweighted graph $G = (V,E)$.
    We construct an ASHG $(N,\vf)$ as follows.
    The set of players is $N = N_V\cup \{z\}$, where $N_V = \{a_u\colon u\in V\}$, i.e., $N_V$ contains a player for each vertex of $G$.
    Symmetric valuations are given by
    $$\vf(i,j) = \begin{cases}
        1 & i = z, j\in N_V\text,\\
        - \negval & i,j\in N_V, \{i,j\}\notin E\text{, and}\\
        0 & i,j\in N_V, \{i,j\}\in E\text.\\
    \end{cases}$$

    Let $\pi$ be a partition of $N$ and let $C_1 = \{u\in V\colon a_u\in \pi(z)\}$.
    Hence, $C_1$ is a vertex set in $G$.
    We create a sequence of vertex sets until we end with a clique of $G$.
    For $i\ge 1$, assume that we have constructed a set $C_i$.
    We stop if $C_i$ is a clique.
    Otherwise, we find a vertex $u_i\in C_i$
    such that $u_i$ is not adjacent to all other vertices in $C_i$ and set $C_{i+1} = C_i\setminus \{u_i\}$.
    Since the number of vertices of $G$ is finite, this stops after $k \le |V|$ steps with a clique $C_{k}$.
    For $1\le \ell\le k$, we define the partition $\pi^{\ell} = \{\{a_u\colon u\in C_{\ell}\}\cup\{z\}\}\cup \{\{a_u\}\colon u\in V\setminus C_{k}\}$.
    We now show that, for $1\le \ell\le k-1$, it holds that $\SW(\pi^{\ell})\le \SW(\pi^{\ell+1})$.
    Indeed, $\vf(u_\ell,u) = 1$ if $u = z$ and $\vf(u_{\ell},u) \le 0$ for all $u\in \pi^{\ell}(u_{\ell})\setminus \{u_{\ell},z\}$.
    Moreover, since $C_{\ell}$ is not a clique, there exists an agent $u\in \pi^{\ell}(u_{\ell})\setminus \{u_{\ell},z\}$ with $\vf(u_{\ell},u) = -\negval \le -1$.
    Hence, removing $u_{\ell}$ from their coalition and forming a singleton coalition can only increase the social welfare.
    In addition, it holds that $\SW(\pi) \le \SW(\pi^1)$ since $\pi^1$ can only differ from $\pi$ by dissolving nonsingleton coalitions not containing $z$, which can only increase the welfare.
    Finally, $\SW(\pi^k) = |C_k|$ as the only nonsingleton coalition in $\pi^k$ is $\pi^k(z)$ which contains exactly $|C_k|$ other agents forming a clique in $G$.
    Hence,
    \begin{equation}\label{eq:ReducedClique}
        \SW(\pi) \le \SW(\pi^k) = |C_k|\text.
    \end{equation}

    Next, let $C^*$ be a maximum clique in $G$.
    Consider $\pi' = \{\{z\}\cup \{a_u\colon u\in C^*\}\}\cup \{\{a_u\}\colon u\in V\setminus C^*\}$ is a partition in $(N,\vf)$ with $\SW(\pi') = |C^*|$.
    Hence, for the partition $\pi^*$ maximizing welfare, it holds that
    \begin{equation}\label{eq:MaxClique}
        \SW(\pi^*) \ge \SW(\pi') = |C^*|\text.
    \end{equation}

    Now assume that we have a polynomial-time algorithm computing a partition $\pi$ with $n^{1-\eps}\cdot \SW(\pi) \geq  \SW(\pi^*)$.
    Clearly, the procedure described above to construct $C_k$ runs in polynomial time as well.
    It holds that
    \begin{equation*}
        n^{1-\eps}\cdot |C_k| \overset{\textnormal{Eq.~(\ref{eq:ReducedClique})}}{\ge} n^{1-\eps}\cdot\SW(\pi) \geq  \SW(\pi^*) \overset{\textnormal{Eq.~(\ref{eq:MaxClique})}}{\ge} |C^*|\text.
    \end{equation*}

    Hence, we have found a polynomial-time algorithm to approximate the maximum clique within a factor of $n^{1-\eps}$.
    As argued in the beginning, this can only happen if \P{} $=$ \NP{}, completing our proof.
\end{proof}

Notably, \Cref{thm:hardnessapprox} immediately implies hardness of $n^{1-\eps}$-\textsc{ApproxWelfare} for ASHGs with nonsymmetric valuations restricted to $\{-1,1\}$.
We can simply replace valuations $\vf(i,j) = 0$ by $\vf_i(j) = 1$ and $\vf_j(i) = -1$ and obtain a reduced instance in which all partitions have the identical welfare.
However, it remains an open problem to resolve the complexity of $c$-\textsc{ApproxWelfare} for symmetric ASHGs with valuations restricted to $\{-1,1\}$, even if $c > 1$ is assumed to be a constant not dependent on $n$.

\subsection{Logarithmic Approximation for Nonnegative Total Value}

We will now show that we can get beyond the inapproximability result of \Cref{thm:hardnessapprox} if we restrict attention to ASHGs with a nonnegative total value.
For this, we will draw a connection to a related problem from the literature on correlation clustering.
Instances of correlation clustering usually only use binary information, i.e., whether two objects should belong to the same or different clusters. 
The goal is then to optimize one of two objectives: maximizing agreements, i.e., the number of pairs whose pairwise relationship is classified correctly, and minimizing disagreements, i.e., the number of pairs classified incorrectly.
In addition, one can consider the combination of these two objectives, where agreements should be maximized while disagreements should simultaneously be minimized.
In the spirit of hedonic games, we capture a weighted version of this objective as a notion of welfare.
Given an ASHG $(N,\vf)$ and a partition $\pi$, we define its \emph{correlation welfare} as
$$\CW(\pi) := \frac{1}{2}\left[\sum_{i\in N} \left(\sum_{j\in \pi(i)}\vf_i(j) - \sum_{j\in N\setminus \pi(i)}\vf_i(j) \right)\right]\text.$$

Charikar and Wirth \cite{ChaWi04} present a randomized $\orderof(\log n)$-approximation algorithm for maximizing\footnote{In their terminology, this is the problem of maximizing \textsc{MaxQP}.} $\CW$ subject to $\pi\in \twoparts$.
They then show that this extends to maximizing $\CW$ within $\Pi_N$ in the case of valuation functions in the range $\{-1,1\}$.
It is easy to see that the same technique applies for the general range of valuation functions (cf. \Cref{lem:opt_2}).
The goal of this section is to extend the approximation guarantee to maximizing $\SW$ for ASHGs with nonnegative total value.

First, by plugging in definitions, we immediately obtain the following relationship between $\SW$ and $\CW$.

\begin{restatable}{proposition}{wfequation}\label{prop:wf-relationship}
    Consider an ASHG $(N,\vf)$ and a partition $\pi$.
    Then it holds that $\CW(\pi) + \frac 12 \TV(N,\vf) = \SW(\pi)$.
\end{restatable}

\begin{proof}
    Let $(N,\vf)$ be an ASHG together with a partition $\pi\in\Pi_N$.
    Then,
    \begin{align*}
        &\CW(\pi) + \frac 12 \TV(N,\vf) \\
        & = \frac{1}{2}\left[\sum_{i\in N} \left(\sum_{j\in \pi(i)}\vf_i(j) - \sum_{j\in N\setminus \pi(i)}\vf_i(j) \right)\right] + \frac 12 \sum_{i\in N} \sum_{j\in N}\vf_i(j)\\
        & = \sum_{i\in N} \sum_{j\in \pi(i)}\vf_i(j) = \SW(\pi)\text.
    \end{align*}
    This proves the desired equation.
\end{proof}

As a consequence, we obtain that the same partitions maximize $\SW$ and $\CW$.

\begin{restatable}{proposition}{maxsame}\label{prop:maxsame}
    Consider an ASHG $(N,\vf)$.
    Then a partition maximizes $\SW$ if and only if it maximizes $\CW$.
\end{restatable}
\begin{proof}
    Consider two partitions $\pi$ and $\pi'$.
    By \Cref{prop:wf-relationship}, it holds that 
    \begin{align*}
        &\SW(\pi) - \SW(\pi')\\
        & = \CW(\pi) + \frac 12\TV(N,\vf) - \CW(\pi') - \frac 12 \TV(N,\vf)\\
        &= \CW(\pi) - \CW(\pi')\text.
    \end{align*}
    Hence, it holds that $\pi \in \argmax_{\pi'\in \Pi_N} \SW(\pi')$ if and only if $\pi \in \argmax_{\pi'\in \Pi_N} \CW(\pi')$.
\end{proof}

Hence, solving for social welfare maximization and correlation welfare maximization is exactly equivalent.
However, this does not have any implications on approximation guarantees, as we illustrate in the next example.

\begin{example}\label{ex:SWvsCW}
    Let $x > 0$ be an arbitrary positive number. 
    Consider the symmetric ASHG $(N,\vf)$ with $N = \{a_1,a_2,a_3\}$ and symmetric valuation functions given by $\vf(a_1,a_2) = 1$, $\vf(a_1,a_3) = -x$, and $\vf(a_2,a_3) = 0$.
    Let $\pi$ denote the singleton partition and $\pi^* = \{\{a_1,a_2\},\{a_3\}\}$.
    Clearly, $\pi^*$ is the unique partition maximizing $\SW$ and $\CW$.
    In addition, it holds that $ \frac{1+x}{x-1} \CW(\pi) = \CW(\pi^*)$, whereas $\frac{\SW(\pi)}{\SW(\pi^*)} = 0$.
    
    Now, consider any approximation guarantee $c > 1$.
    Then, $\pi$ provides a $c$-approximation of $\CW$ for $x = \frac {c - 1}{1 + c}$, whereas $\pi$ yields no $c$-approximation of $\SW$.
    \hfill$\lhd$
\end{example}

The reason why approximate outcomes in terms of correlation welfare do not yield any guarantee on the social welfare in \Cref{ex:SWvsCW} is that there is a single valuation that is very negative. 
It is enough that the two corresponding agents are in different coalitions to obtain a good approximation of the correlation welfare. 
However, the picture changes if such a situation does not occur.
If the total value of an instance is nonnegative, then social welfare inherits approximation guarantees from correlation welfare.

\begin{restatable}{lemma}{welfaretransfer}\label{lem:welfaretransfer}
    Let $(N,\vf)$ be an ASHG such that $\TV(N,\vf)\ge 0$. 
    Let $c\ge 1$ and let $\pi^*$ be a partition maximizing $\CW$. 
    Let $\pi$ be a partition with $c\cdot\CW(\pi) \ge \CW(\pi^*)$.
    Then it holds that $c\cdot\SW(\pi) \ge \SW(\pi^*)$.
\end{restatable}
\begin{proof}
    Consider an ASHG $(N,\vf)$ such that $\TV(N,\vf)\ge 0$. 
    Let $c\le 1$ and let $\pi^*$ be a partition maximizing $\CW$. 
    Consider a partition $\pi$ with $\CW(\pi) \ge c\cdot\CW(\pi^*)$.

    Then it holds that
    \begin{align*}
        c\cdot\SW(\pi) &\overset{\textnormal{\Cref{prop:wf-relationship}}}{=} c\cdot\CW(\pi) + \frac c2\cdot\TV(N,\vf)\\
        &\ge \CW(\pi^*) + \frac c2\cdot\TV(N,\vf)\\
        &\overset{\textnormal{\Cref{prop:wf-relationship}}}{=} \SW(\pi^*) - \frac 12\TV(N,\vf) + \frac c2\cdot\TV(N,\vf)\\
        &= \SW(\pi^*) + \frac 12\left(c-1\right)\TV(N,\vf)\\
        &\ge  \SW(\pi^*)\text.
    \end{align*}
    In the last inequality, we used that $\TV(N,\vf)\ge 0$ and $c\ge 1$.
\end{proof}

As a second lemma, we establish a relationship between maximizing welfare when partitions can only contain two coalitions and when partitions are unconstrained.
Its proof is an adaptation of a similar result by Charikar and Wirth \cite{ChaWi04} concerning $\CW$ for valuations in $\{-1,1\}$.

\begin{restatable}{lemma}{OPTtwo}\label{lem:opt_2}
    Let $(N,\vf)$ be an ASHG with $\TV(N,\vf)\ge 0$.
    Then it holds that $\max_{\pi\in \Pi_N}\SW(\pi) \leq 2 \cdot \max_{\pi\in \twoparts}\SW(\pi)$.
\end{restatable}
\begin{proof}
    Consider an ASHG $(N,\vf)$ with $\TV(N,\vf)\ge 0$.
    Let $\pi^*$ be a partition maximizing $\SW$.
    Define the quantities
    \begin{itemize}
        \item $W := \sum_{i\in N, j\in \pi^*(i)} \vf_i(j)$,
        \item $A := \sum_{i\in N, j\notin \pi^*(i)} \vf_i(j)$,
    \end{itemize}
    
    These measure
    the total value of within and across coalitions in the optimal partition $\pi^*$. 
    By definition, it holds that
    \[
    \SW(\pi^*) = W \quad \text{and} \quad \TV(N,\vf) = W + A\text.
    \]
    We now establish a lower bound for the optimal solution welfare of partitions with two coalitions.
    Therefore, define
    $\pi^*_2 := \argmax_{\pi\in \twoparts}\SW(\pi)$.
    While the exact social welfare of $\pi^*_2$ is not computable, we can establish a lower bound for it. 
    Consider the random partition where each coalition in $\pi^*$ is assigned to one of two new coalitions, $A$ and $B$, uniformly at random. 
    Then the expected social welfare of the random partition $\pi^{(2)} := \{A,B\}$ is a lower bound for the social welfare of $\pi^*_2$. 
    Every pair of agents in a common coalition in $\pi^*$ will remain in a common coalition in $\pi^{(2)}$.
    Moreover, agents that have been in different coalitions in $\pi^*$ are in the same coalition in $\pi^{(2)}$ with probability~$\frac 12$.
    Hence,
    $\EV\left[\SW\left(\pi^{(2)}\right)\right] = W + \frac 12 {A} \leq \pi^*_2$. 
    
    We conclude that 
    \begin{align*}
        \SW(\pi^*) &= W \le W + \TV(N,\vf) = 2W + A\\ 
        &= 2\cdot\EV\left[\pi^{(2)}\right] \leq 2\cdot\SW(\pi^*_2)\text.
    \end{align*}
    There we use that $\TV(N,\vf)\ge 0$. 
\end{proof}

We can combine the two last lemmas to apply the main theorem by Charikar and Wirth \cite{ChaWi04} and obtain a randomized $\orderof(\log n)$-approximation algorithm.

\begin{theorem}
    There exists a randomized $\orderof(\log n)$-approximation algorithm for maximizing social welfare in ASHGs with nonnegative total value.
\end{theorem}

\begin{proof}
    Theorem~1 by Charikar and Wirth \cite{ChaWi04} states the existence of a randomized $\orderof(\log n)$-approximation algorithm for $\CW$ under the constraint that partitions are in $\twoparts$.
    By \Cref{lem:welfaretransfer}, the same approximation guarantee is obtained for $\SW$ under the same constraint.
    Finally, \Cref{lem:opt_2} guarantees that the maximum welfare of any partition is better by at most a factor of~$2$.
\end{proof}

\section{Beyond Worst-Case Analysis}
\label{sec:beyondWC}

In light of the hardness result by Flammini et al. \cite{FKV22a} for approximating social welfare in aversion-to-enemies games, it is natural to ask how well we can approximate welfare in such games generated by stochastic models. 
In this section, we introduce two such models where the valuations originate from either Erd\H{o}s-R\'{e}nyi or multipartite graphs. 
Erd\H{o}s-R\'{e}nyi graphs serve as a common testbed for 
graph optimization problems and help us set the stage for the more challenging setting of multipartite graphs. 
Interestingly, our main theorems demonstrate that greedy algorithms are remarkably effective in these models, yielding constant-factor and logarithmic-factor approximations of social welfare. 

In this section, $G = (N,\vf)$ refers to a fixed symmetric aversion-to-enemies game.
In any partition of $N$, a valuation of $-n$ within a coalition implies a negative utility for the corresponding agents.
Consequently, removing one of these agents from the coalition and forming a singleton coalition would increase the overall social welfare. 
This observation suggests that in an optimal partition $\pi^*$, no coalition contains agents with a mutual valuation of $-n$.
Let $G'$ denote the subgraph of $G$, obtained by removing all edges with weight $-n$. 
We now present a useful lemma.

\begin{lemma}\label{lem:maxclique}
    If the size of the maximum clique in $G'$ is $t$, then $\SW(\pi^*) \leq n(t-1)$. 
\end{lemma}
\begin{proof}
No coalition in the partition \(\pi^*\) contains an edge with weight \(-n\). 
Therefore, each coalition in \(\pi^*\) forms a clique in $G'$. 
Since the size of a maximum clique in \(G'\) is \(t\), the size of every coalition in \(\pi^*\) is at most \(t\). 
Consequently, the utility of each agent is bounded by \(t-1\). 
We conclude that 
\[
\SW(\pi^*) = \sum_{i \in N} u_i(\pi^*(i)) \leq n(t-1).\qedhere
\]
\end{proof}

\subsection{Erd\H{o}s-R\'{e}nyi Graphs}\label{sec:ERgraphs}
In our first model, we assume a set of agents, each pair of which is incompatible with probability~$1-p$.
We model this as a symmetric aversion-to-enemies game by assigning a valuation of $-n$ between incompatible agents and a valuation of $1$ between compatible agents. 
This corresponds to sampling its underlying graph as follows.

\begin{definition}
    A \emph{weighted Erd\H{o}s-R\'enyi graph} $\ER{n}{p}$
    is a random weighted graph with $n$ vertices such that, independently, 
    each edge takes a weight of \(-n\) with probability \(p\) and a weight of \(1\) with probability \(1-p\).  
\end{definition}

We will show that a simple and natural greedy algorithm yields a constant-factor approximation of the maximum welfare with high probability. 
For this, we use the {\em greedy clique formation algorithm} in Section~5.2 of \cite{BuKr24a} applied to the subgraph $G'$ formed by removing all negative edges from a graph $G$. 
The algorithm greedily forms maximal cliques in \(G'\).
as long as the cliques reach a certain size threshold \(t = \left\lceil\frac{\log_{1/p} n}{2} \right\rceil\). 
If, at any point, the size of the created maximal clique is smaller than \(t\), the algorithm outputs the existing cliques as coalitions, and assigns singleton coalitions to the remaining agents.
The following theorem measures the performance of this algorithm.
It follows from the proof of Theorem~5.2 by Bullinger and Kraiczy \cite{BuKr24a}.\footnote{We remark that in their model the edges in the cliques occur with probability $p$ whereas they occur with probability $1-p$ in our model.}

\begin{theorem}[\cite{BuKr24a}]\label{thm:prevpaper}
     Consider an Erd\H{o}s-R\'enyi graph $G=(n, p)$ and let $b = \frac 1{p}$. 
     Then, with probability at least $1 - e^{-\Omega\left(\log_b^3 n\right)}$, the greedy clique formation algorithm assigns all except at most $\frac n{\log_b^2 n}$ to cliques of size $\left\lceil\frac{\log_{b} n}{2} \right\rceil$.
\end{theorem}

We apply the theorem to obtain a constant-factor approximation of maximum welfare.
Essentially, \Cref{thm:prevpaper} allows to obtain a coalition with social welfare $\Theta(n\log n)$ while we apply \Cref{lem:maxclique} to show that the maximum welfare is also of this order.

    \begin{theorem}
        Let $p\in (0,1)$. 
        Then there exists a constant-factor approximation algorithm for aversion-to-enemies games given by a weighted Erd\H{o}s-R\'enyi graph $G =(n, p)$.
    \end{theorem}

\begin{proof}
    Note that $G'$ is a weighted Erd\H{o}s-R\'enyi graph where every edge was sampled with probability $p$, i.e. $G'=(n,p)$.
    Let $b = \frac{1}{p}$ and let \(\pi\) be the resulting partition after applying the greedy clique formation algorithm to the subgraph $G'$.
    By \Cref{thm:prevpaper}, it follows that $\EV[\SW(\pi)] = \Theta(n\log n)$.

    In addition, the size of the maximum clique in $G'$ is $\mathcal O(\log_b n)$ with probability~$1$ \cite{GrMc75a}.
    Hence, by \Cref{lem:maxclique},
    the expected maximum welfare of a partition in $G'$ is $\mathcal O(n\log n)$.
    Note that all constants hidden in the asymptotic behavior only depend on $b$ and, therefore, on~$p$.
    It follows that the greedy clique formation algorithm yields a constant-factor approximation.
\end{proof}

To get a feeling on the constant hidden in the previous theorem, one can reason as follows. 
For large enough $n$, with probability $1- \frac 1n$, greedy clique formation will place $\frac {2n}3$ agents into coalitions of size $\frac {\log_b n}2$, resulting in an expected social welfare of at least $(1-\frac 1n)\frac {2n}3\frac {\log_b n}2 \ge \frac {n \log_b n}4$.
Moreover, for large enough $n$, with probability $1-\frac 1{n^2}$, the maximum clique is of size at most $4\log_{\frac 1p} n$ \cite{GrMc75a}.
Hence the expected maximum welfare is at most $n(4\log_{\frac 1{1-p}} n + \frac 1{n^2}n^2) = n(4\log_{\frac 1{1-p}} n + 1)$, where the second term bounds the maximum welfare of a partition with the formation of a clique of all agents for the remaining cases occurring with probability at most$\frac 1{n^2}$. 
This yields a ratio of about $\frac 1{16} \log_{\frac 1{1-p}} b$.
 
\subsection{Random Multipartite Graphs}
Consider \( k\ge 2 \) distinct groups of agents, where the goal is to form diverse coalitions that contain at most one agent from each group. 
We model this by assigning a negative edge weight of \( -n \) to any pair of agents within the same group, rendering them incompatible.
Additionally, certain pairs of agents from different groups may also be incompatible. 
In a random $k$-partite graph, any pair of agents from different groups is incompatible with probability \( p \). This problem can be formalized as follows.

Consider a $k$-partite graph where vertices represent agents. 
The graph consists of \( n \) vertices partitioned into \( k \) disjoint ``color'' classes \( V_1, V_2, \dots, V_k \). 
All our results hold if $k$ is either a constant or any function satisfying  $k=o\left(\frac{n}{\log n}\right)$. 
Without loss of generality, assume that the color classes are sorted in nonincreasing order by the number of vertices they contain, i.e., \( |V_1| \geq |V_2| \geq \cdots \geq |V_k| \).

    A $k$-partite graph is said to be \emph{balanced} if \(|V_k| \geq q|V_1|\) holds for some constant \(q\in(0,1)\).
    A \emph{Tur\'{a}n graph} 
    is a special case of a balanced $k$-partite graph with $n$ vertices, where each color class contains the same number of vertices, i.e., for all $i \in [k]$, we require $|V_i| = \frac{n}{k}$.     

We capture these in our second model of random graphs inducing aversion-to-enemies games.

\begin{definition}\label{def:rule}
    A \emph{random balanced \( k \)-partite graph} \( G = (\{V_1, \ldots, V_k\}, p) \) is a weighted graph  
    where edge weights are sampled independently as follows: 
    each edge between vertices in two different color classes independently takes a weight of \(-n\) with probability $p$, and a weight of $1$ with probability $1-p$; each edge between vertices of the same color class takes a weight of \(-n\) with probability~$1$.
    The input parameter \( p \in (0,1) \) is called the \emph{perturbation probability}, and it is allowed to depend on~$n$. 

    A \emph{random Tur\'{a}n graph} $G=(n,k,p)$ is a random balanced $k$-partite graph where each color class contains the same number of vertices, i.e., for all $i \in [k]$ we have $|V_i| = \frac{n}{k}$.   
 \end{definition}

The goal is to find a partition of maximum welfare for the case when the input is a random balanced $k$-partite graph (or a random Tur\'{a}n graph). 
Note that when $p=0$ or $p=1$, the problem becomes trivial:
When $p = 0$, all weights between vertices from different color classes are deterministically positive, and the graph $G'$ induced by edges of weight~$1$ is a complete $k$-partite graph. 
In this case, an optimal partition of a Tur\'{a}n graph consists of $\frac{n}{k}$ coalitions, each containing a unique member from each of the color classes \( V_1, \dots, V_k \). 
For a general balanced $k$-partite graph, the welfare-maximizing partition contains $|V_k|$ $k$-cliques, $|V_{k-1}|-|V_k|$ $(k-1)$-cliques, etc.
Conversely, when $p=1$, then all edges in the graph have weight of $-n$, which implies the maximum welfare is obtained by the singleton partition. 

We now establish a straightforward upper bound on the maximum welfare. 
Recall that $G'$ is the graph obtained by removing all negative edges from $G$. 
Since $G'$ is $k$-partite, the maximum clique size in $G'$ is at most $k$. Thus, \Cref{lem:maxclique} implies the following proposition. 
\begin{proposition}\label{prop:n(k-1)}
    In a random balanced $k$-partite graph, the maximum welfare is bounded by $\SW(\pi^*) \leq n(k-1)$.
\end{proposition}

In our analysis, both $k$ and $p$ can depend on $n$. 
We now present polynomial-time algorithms that compute a constant-factor approximation of social welfare when \( p = \orderof\left(\frac{1}{k}\right)\), and a \( \log_e n \)-approximation when \( p \) is a constant for random balanced \(k\)-partite graphs. 
We first illustrate our results by providing proofs for the special case of random Tur\'{a}n graphs.
We then extend our results to random balanced $k$-partite graphs in \Cref{sec:balanced}.
They are obtained by employing a reduction to Tur\'{a}n graphs.

\subsubsection{Low Perturbation Regime for Random Tur\'{a}n Graphs}
\Cref{alg:efficientgreedy} takes as input a random Tur\'{a}n graph, and an accuracy parameter (constant number) $\eps > 0$. In addition, the algorithm takes as input a subset of color classes denoted by \( S \subseteq \{V_1, \dots, V_k\} \). The number of color classes in \( S \) is denoted by $k' \le k$. Our algorithm outputs a partition of the vertices in \( S \) into coalitions, meaning that only the vertices in \( \bigcup_{V_i \in S} V_i \) are considered, as if the graph consisted of exactly \( k' \) color classes, each containing \( \frac{n}{k} \) vertices.

\Cref{alg:efficientgreedy} begins by randomly selecting a vertex to initiate the formation of the first coalition. 
It then iteratively adds a new vertex $w$ to the coalition if $w$ has only edges of weight \( 1 \) towards all current members of the coalition (this ensures there are no $-n$ edges in the created coalitions). 
This process continues until no additional vertices can be included. 
Hence, each formed coalition is a maximal clique in the subgraph \( G' \) obtained by removing all negative edges, and we refer to these coalitions as maximal cliques. 
If the size of the resulting maximal clique exceeds \( k'\sqrt{1-\eps} \), the vertices in the clique are removed from the pool of available vertices, and the process is repeated with the remaining vertices. 
However, if at any point the size of the obtained maximal clique is smaller than \( k'\sqrt{1-\eps} \), the algorithm terminates and returns the current set of coalitions, with any remaining vertices assigned to singletons coalitions.

\begin{algorithm}
\caption{Greedy coalition formation}\label{alg:efficientgreedy}
\begin{flushleft}
\textbf{Input:} $\langle G, S, \eps\rangle$ where $G=(n,k,p)$ is a random Tur\'{a}n graph, $S\subseteq \{V_1,\cdots, V_k\}$ is a subset of color classes with $|S|=k'$, and $\eps \in (0,1)$.\\
\textbf{Output:} Partition $\pi$ on $\bigcup_{V_i \in S} V_i$
\end{flushleft}
\begin{algorithmic}[1]
\State $\pi \gets \emptyset$, $R \gets \bigcup_{V_i \in S} V_i$
\While{$R\neq \emptyset$}
    \State Select $v \in R$ to begin coalition $C= \{v\}$
    \State $L \gets R$
    \While{$\exists w\in L$ with all edges towards $C$ of weight 1}
            \State $C \gets C \cup \{w\}$, $L \gets L \setminus \{w\}$
    \EndWhile
    \If{$|C| \ge k'(\sqrt{1-\eps})$}
        \State  $\pi \gets \pi \cup \{C\}$, $R \gets R \setminus C$
    \Else \State \Return $\pi \cup \{\{v\}\colon v\in R\}$
    \EndIf
\EndWhile
\State \Return $\pi$ 
\end{algorithmic}
\end{algorithm}

The following lemma shows that for sufficiently small values of \( p \), by selecting a subset of color classes \( S \subseteq \{V_1, \dots, V_k\} \), where \( k' = |S| \), and running \Cref{alg:efficientgreedy}, the algorithm produces nearly \( \frac{n}{k} \) maximal cliques, each of which exceeds the size \( k'\sqrt{1-\eps} \) with high probability. When \( p = \orderof\left(\frac{1}{k}\right) \), the input set \( S \) can include all \( k \) color classes, i.e., \( S = \{V_1, \dots, V_k\} \). In this case, with high probability, the algorithm finds nearly \( \frac{n}{k} \) maximal cliques, each larger than \( k\sqrt{1-\eps} \), making the size of these cliques nearly identical to the ideal clique size of \( k \) for small values of $\eps$. Consequently, this results in a constant approximation of the social welfare.

\begin{lemma}\label{lem:lowerbound}
    Consider a random Tur\'{a}n graph $G=(n,k,p)$, and a nonempty subset of color classes $S\subseteq \{V_1,\cdots,V_k\}$, and $p=\orderof(\frac{1}{k'})$ for $k'= |S|$. 
    For any fixed  $\eps \in (0,1)$ and $\alpha \in (0,1)$, \Cref{alg:efficientgreedy} returns a partition $\pi$ with $\alpha \frac{n}{k}$ cliques of size at least $k'\sqrt{1-\eps}$ with probability $1-e^{-\Theta\left(\frac{nk'}{k}\right)}$. 
\end{lemma}
\begin{proof}

We prove that the size of the first $\frac nk$ maximal cliques exceeds \( k'(\sqrt{1-\eps}) \) with high probability. 
Let $C$ denote a clique that is formed during the $i$th iteration of the while loop, with a current size of \( t \). 
A color class is said to be \emph{available} if no vertex from that class has been added to $C$, yet. 
The probability that \( C \) is maximal is \( \left(1 - (1 - p)^t\right)^{(k' - t)\left(\frac{n}{k} - i\right)} \), as this represents the probability that none of the \( \frac{n}{k} - i \) remaining vertices in each of the \( k' - t \) available color classes can be added to \( C \), due to having at least one edge of weight \(-n\) with a vertex in \( C \).

Let \( X \) denote the number of maximal cliques in the subgraph induced by the color classes in $S$, where each clique has size at most \( t_0 = k' \sqrt{1 - \eps} \). As there are \( \binom{k'}{t} (\frac{n}{k})^t \) cliques of size \( t \), and since \( \binom{k'}{t} (\frac{n}{k})^t \leq \binom{k}{t} (\frac{n}{k})^t \leq k^t (\frac{n}{k})^t \leq n^t \), the following holds:

\begin{align*}
        & \EV[X] \leq \sum_{t=1}^{t_0} n^t     \left(1 - (1 - p)^t\right)^{(k' - t)\left(\frac{n}{k} - i\right)} \\
        & \leq t_0 n^{t_0} \left(1 - (1 - p)^{t_0}\right)^{(k' - t_0)\left(\frac{n}{k} - i\right)} \\
        & = t_0 n^{t_0} (1 - (1 - p)^{t_0})^{n(\frac{k'-t_0}{k}-\frac{(k'-t_0)i}{n})} \\
        & \leq t_0 n^{t_0} e^{-(1-p)^{t_0}\left[n(\frac{k'-t_0}{k}-\frac{(k'-t_0)i}{n})\right]} \\
        & = t_0 e^{t_0 \log_e{n}} \cdot e^{-(1-p)^{t_0}\left[n(\frac{k'-t_0}{k}-\frac{(k'-t_0)i}{n})\right]}
\end{align*}
where we used $1-x \leq e^{-x}$. 
Since $t_0 = k'(\sqrt{1-\eps})$ and $p=\orderof(\frac{1}{k'})$, $(1-p)^{t_0} \geq e^{t_0(-p-p^2)} \to c_0$ for some positive constant $c_0$. Therefore, the expression can be rewritten as:
\begin{align*}
        & \EV[X] \leq t_0 e^{t_0 \log_e{n}}\cdot e^{c_0\left[-n(\frac{k'-t_0}{k}-\frac{(k'-t_0)i}{n})\right]} \\
        & = k'(\sqrt{1-\eps}) e^{k'\sqrt{1-\eps}\log_e{n}} \cdot
        e^{c_0\left[-n\left(\frac{k'}{k}(1-\sqrt{1-\eps}) - \frac{k'(1-\sqrt{1-\eps})i}{n}\right)\right]}.
\end{align*}

For any constant $\alpha \in (0,1)$,  while the current iteration satisfies $i \leq \alpha \frac{n}{k}$,  it holds that
\begin{align*}
        & \EV[X] \leq  k'(\sqrt{1-\eps}) e^{k'(\sqrt{1-\eps}) \log_e{n}} 
        e^{c_0\left[-n\frac{k'}{k}(1-\alpha)(1-\sqrt{1-\eps})\right]}.
\end{align*}

Let $a_0=\sqrt{1-\eps}$ and $b_0 = c_0(1-\alpha) (1-\sqrt{1-\eps})$ be two constant numbers, 
\begin{equation*}
    \EV[X] \leq k' e^{k'[a_0 \log_e n -b_0 \frac{n}{k}]}\text.
\end{equation*}

Since \( k = o\left(\tfrac{n}{\log n}\right) \),  we have that \( \EV[X] \) tends to zero as \( n \) tends to infinity. 
By Markov's inequality, the probability of having at least one maximal clique of size at most \( t_0 \) is at most \( k' e^{-\Theta\left(\frac{nk'}{k}\right)} \).

Thus, the probability of exiting the first while loop during the \(i\)th iteration, where \( i \leq \alpha \frac{n}{k} \), is also at most \( k' e^{- \Theta\left(\frac{nk'}{k}\right)} \). 
By a union bound, the probability that the algorithm exits the while loop before \( i = \alpha \frac{n}{k} \) is bounded by
\[
\alpha\frac{n}{k} k' e^{- \Theta\left(\frac{nk'}{k}\right)} = e^{- \Theta\left(\frac{nk'}{k}\right)}.
\]

Therefore, with probability \( 1 - e^{- \Theta\left(\frac{nk'}{k}\right)} \), the algorithm returns \( \alpha \frac{n}{k} \) cliques of size \( k' (\sqrt{1-\eps}) \).
\end{proof}

We now prove our main theorem for the low perturbation regime.

\begin{theorem}
    Consider aversion-to-enemies games given by random Tur\'{a}n graphs $G = (n,k,p)$, where $k\ge 2$ and $p=\orderof(\frac{1}{k})$. 
    Then there is a polynomial-time algorithm that returns a constant-factor approximation to maximum welfare with probability $1-e^{-\Theta(n)}$. 
\end{theorem}

\begin{proof}
Fix \( \eps \in (0,1) \), and consider \Cref{alg:efficientgreedy} for input $\langle G, S = \{V_1, \dots, V_k\}, \eps\rangle$. 
Since \( p = \orderof\left(\frac{1}{k}\right) \), \Cref{lem:lowerbound} implies that the algorithm returns \( \alpha \frac{n}{k} \) cliques of size at least \( k(\sqrt{1-\eps}) \) with probability \( 1 - e^{-\Theta(n)} \), where \( \alpha \) is any constant in the range \( (0,1) \). Each clique contains at least \( k \sqrt{1-\eps} \) agents, and the utility of every agent in such cliques is at least \( k(\sqrt{1-\eps}) - 1 \). Therefore, the social welfare of the partition returned by the algorithm is bounded as
$$\SW(\pi) \geq \alpha \frac{n}{k} k\sqrt{1-\eps}(k(\sqrt{1-\eps})-1) = \alpha n k (1-\eps) - \alpha  n\sqrt{1-\eps}\text.$$

Moreover, \Cref{prop:n(k-1)} implies that $\SW(\pi^*) \leq n(k-1) < nk$. 
Hence,
\begin{equation}\label{eq:approxTuran}
    \frac{\SW(\pi)}{\SW(\pi^*)} \geq \alpha (1-\eps) - \frac{\alpha \sqrt{1-\eps}}{k}\text.
\end{equation}

Note that $\eps$ is a parameter of our choice, and it can be chosen arbitrarily close to zero.
In addition, $\alpha$ is a constant with $\alpha \in (0,1)$, meaning that $\alpha$ can be made arbitrarily close to $1$. 
This implies that the approximation factor as bounded in \Cref{eq:approxTuran} can be arbitrarily close to $1-\frac{1}{k} \ge \frac 12$, where we use that $k\ge 2$. 
\end{proof}

As we just showed, the approximation factor can be arbitrarily close to $1-\frac{1}{k}$.
Hence, in case that $k$ tends to infinity as $n$ tends to infinity, we obtain nearly optimal partitions for large $n$.

\subsubsection{High Perturbation Regime for Random Tur\'{a}n Graphs}

We now present a second algorithm, which uses \Cref{alg:efficientgreedy} as a subroutine, for the case when the perturbation probability is constant, i.e., \( p = c \) for some constant \( c \in (0,1) \).
\Cref{alg:efficientgreedy2} partitions the set of color classes $\{V_1,\cdots, V_k\}$ into \( \lceil ck \rceil \) disjoint sets, so that the sizes of the sets differ by at most one. Denote these disjoint sets by $\{S_1, \dots, S_{\lceil ck \rceil} \}$.\footnote{
For example, if $k=11$ and $c=\frac{1}{3}$, one can take $S_1 = \{V_1,V_2,V_3\}$, $S_2 = \{V_4,V_5,V_6\}, S_3 = \{V_7,V_8,V_9\},$ and $S_4 = \{V_{10},V_{11}\}$. }

\begin{lemma}\label{lem:alg2bound}
    Consider a random Tur\'{a}n graph $G=(n,k,p)$ where $p = c$ for some constant $c$.
    Let $\pi$ be the partition returned by \Cref{alg:efficientgreedy2}. 
    Then $\SW(\pi) = \Omega(n) $ with probability $1-ne^{-\Theta(\frac{n}{k})}$.
\end{lemma}
\begin{proof}
At least \( \lfloor ck \rfloor \) of the sets in $\{S_1, \dots, S_{\lceil ck \rceil} \}$ contain \( k' = \left\lfloor \frac{1}{c} \right\rfloor \) color classes. Since \( k' \) is a constant, it follows that \( p = \orderof\left(\frac{1}{k'}\right) \). 
Each set of colors forms a subproblem, and by applying \Cref{alg:efficientgreedy} to each subproblem, we obtain \( \alpha \frac{n}{k} \) coalitions of size \( k' \sqrt{1-\eps} \) with probability \( 1 - e^{-\Theta\left(\frac{n}{k}\right)} \), for some fixed constants \( \eps \in (0,1) \) and \( \alpha \in (0,1) \), as established by \Cref{lem:lowerbound}. 
By a union bound, the probability that \Cref{alg:efficientgreedy} returns fewer than \( \alpha \frac{n}{k} \) coalitions of size \( k' \sqrt{1-\eps} \) in at least one subproblem is at most \( \lceil ck \rceil e^{-\Theta\left(\frac{n}{k}\right)} \leq n e^{-\Theta\left(\frac{n}{k}\right)} \).
Since \( k = o\left(\frac{n}{\log n}\right) \), this probability approaches zero as \( n \) increases. 
Therefore, with probability \( 1 - n e^{-\Theta\left(\frac{n}{k}\right)} \), all subproblems return at least \( \alpha \frac{n}{k} \) coalitions of size \( k' \sqrt{1-\eps} \).
    The utility of agents in these coalitions is least $k'\sqrt{1-\eps}-1$.
    Therefore,
    $$
    \SW(\pi) \geq \lfloor kc \rfloor \alpha \frac{n}{k} (k'\sqrt{1-\eps}) (k'\sqrt{1-\eps}-1)\text. 
    $$
Since $k' = \left\lfloor \frac{1}{c} \right\rfloor$ and $c$ is constant, it follows that $\SW(\pi) \geq n c_0$ for some constant~$c_0$. 
\end{proof}

\begin{algorithm}
\caption{Dividing into smaller subproblems}\label{alg:efficientgreedy2}
\begin{flushleft}
    \textbf{Input:} $\langle G,\eps\rangle$ where $G=(n,k,p)$ is a random Tur\'{a}n graph and $p=c$ for some constant $c\in (0,1)$ \\
    \textbf{Output:} Partition $\pi$
\end{flushleft}
\begin{algorithmic}[1]
\State $\pi \leftarrow \emptyset$
\State Partition \( \{V_1, \dots, V_k\} \) into \( \lceil ck \rceil \) disjoint sets \( S_1, \dots, S_{\lceil ck \rceil} \) that differ in size by at most one
\For{each group of colors $S\in \{S_1,\cdots, S_{\lceil kc\rceil}\}$}
    \State Let $\pi_S$ be the partition within vertices in color classes of $S$, after applying \Cref{alg:efficientgreedy} on input $\langle G, S, \eps\rangle$ 
    \State $\pi \gets \pi \cup \pi_S$
\EndFor
\State \Return $\pi$ 
\end{algorithmic}
\end{algorithm}

We now bound the maximum welfare.
The proof is similar to our arguments for Erd\H{o}s-R\'{e}nyi graphs.

\begin{restatable}{lemma}{LemLargeKSWOPT}\label{lem: largeKSWOPT}
    If \( k = \Omega(\log n) \) and \( p = c \) for some constant \( c \in (0,1) \), then the maximum welfare satisfies \( SW(\pi^*) = \orderof(n \log n) \) with probability \( 1 - \left( \frac{ce}{2\log_e n} \right)^{\frac{2}{c} \log_e n} \).
\end{restatable}
\begin{proof}
Recall that \( G' \) is the graph obtained from \( G \) by removing all edges with weight \(-n\). Define \( t = \frac{2}{c} \log_e n + 1 \), and let \( S^t \) represent the set of all groups of \( t \) agents, each from a different color class. It is easy to observe that \( |S^t| = \binom{k}{t} (\frac{n}{k})^t \). For any \( s \in S^t \), let \( E_s \) denote the event that all vertices in \( s \) form a clique in \( G' \), so the probability of this event is \( P(E_s) = (1-p)^{\binom{t}{2}} \). Let \( X^t \) denote the number of cliques of size \( t \) in \( G' \).

    \begin{equation}
        \begin{aligned}
        \EV[X^t] & = \sum_{s \in S^{t}} P(E_s) = \binom{k}{t} (\frac{n}{k})^{t} (1-p)^{\binom{t}{2}} \\
        & \leq (\frac{ek}{t})^{t}(\frac{n}{k})^{t} (1-p)^{\binom{t}{2}} \\
        & = (\frac{en}{t} (1-p)^{\frac{t-1}{2}})^{t} 
        \\& \leq \left(\frac{en}{t} e^{-p\frac{t-1}{2}}\right)^{t}
        \end{aligned}
    \end{equation}
    where we used $1-p \leq e^{-p}$ when $p\in[0,1]$ in the last inequality. For a choice of $t=\frac{2}{p}\log_e n+1$ 
   \begin{equation}
        \begin{aligned}
        \EV[X^t] 
    & \leq \left(\frac{e}{\frac{2}{p}\log_e n+1}\right)^{\frac{2}{p}\log_e n+1}.
        \end{aligned}
    \end{equation}
By Markov's inequality, the probability of having at least one clique of size \( t \) in \( G' \) is at most \( \left(\frac{ce}{2\log_e n}\right)^{\frac{2}{c} \log_e n} \), where \( p = c \). This implies that, with high probability, \( G' \) contains no clique of size \( \frac{2}{c} \log_e n + 1 \). Therefore, the largest clique in \( G' \) is smaller than \( \frac{2}{c} \log_e n + 1 \) with probability \( 1 - \left(\frac{ce}{2\log_e n}\right)^{\frac{2}{c} \log_e n} \). According to \Cref{lem:maxclique}, the maximum welfare is upper-bounded as follows:
$$\SW(\pi^*) \leq n(t-1) = \frac{2n}{c} \log_e n = \orderof(n\log_e n).$$ 
This completes the proof.
\end{proof}

We now prove our main result for the high perturbation regime.
Again, the theorem extends to random balanced $k$-partite graphs. 

\begin{theorem}
    Consider aversion-to-enemies games given by random Tur\'{a}n graphs $G = (n,k,p)$,
    where \( p = c \) for some constant \( c \in (0,1) \). 
    Then there exists a polynomial-time algorithm that returns a partition that provides a $\orderof(\log n)$-approximation of the maximum welfare with probability \( 1 - ne^{-\Theta\left(\frac{n}{k}\right)} - \left( \frac{ce}{2\log_e n} \right)^{\frac{2}{c} \log_e n + 1} \).
\end{theorem}

\begin{proof}
\Cref{lem:alg2bound} implies that \Cref{alg:efficientgreedy2} returns a partition \( \pi \) where \( \SW(\pi) = \Omega(n) \) with probability \( 1 - ne^{-\Theta\left(\frac{n}{k}\right)} \). A simple upper bound on the maximum welfare is provided in \Cref{prop:n(k-1)}, which implies \( \SW(\pi^*) \leq n(k-1) \). If \( k = \orderof(\log n) \), this results in \( \SW(\pi^*) = \orderof(n \log n) \). However, when \( k = \Omega(\log n) \), it does not provide a useful guarantee. Instead, \Cref{lem: largeKSWOPT} shows that even in this case, \( \SW(\pi^*) = \orderof(n \log n) \) with probability \( 1 - \left( \frac{ce}{2\log_e n} \right)^{\frac{2}{c} \log_e n} \). 
By a union bound, we have 
\[
\SW(\pi) \geq \Omega\left(\frac{1}{\log_e n}\right) \SW(\pi^*).
\]
with probability \( 1 - ne^{-\Theta\left(\frac{n}{k}\right)} - \left( \frac{ce}{2\log_e n} \right)^{\frac{2}{c} \log_e n + 1}\).
\end{proof}

\subsection{Balanced Multipartite Graphs}\label{sec:balanced}

In this section, we show how to extend our results for random Tur\'{a}n graphs to random balanced multipartite graphs. The main idea is a reduction to the case of Tur\'{a}n graphs. 
Recall that in a balanced $k$-partite graph, the input has components of comparable sizes, i.e., \(|V_k| \geq q|V_1|\) holds for some constant \(q\in(0,1)\). 
We refine the greedy algorithm by forcing the components' sizes to become equal.
This simple procedure is outlined in \Cref{alg:balancedgreedy-phase1}.

\begin{algorithm}
\caption{Reduction to a random Tur\'{a}n graph}\label{alg:balancedgreedy-phase1}
\begin{flushleft}
    \textbf{Input:} $\langle G,\eps\rangle$ where $G=(\{V_1, \cdots, V_k\},p)$ is a random balanced $k$-partite graph\\
    \textbf{Output:} A random Tur\'{a}n graph $G^T$
\end{flushleft}
\begin{algorithmic}[1]
\For{each color class \( V_i \) with \( i \in [k] \)}
{
    Select an arbitrary subset of agents \( V'_i \subseteq V_i \) such that \( |V'_i| = |V_k| \)
}
\EndFor
\State Let \( G^T = (n', k, p) \) denote the Tur\'{a}n graph induced by the vertices in \( \bigcup_{i=1}^{k} V'_i \) where $n'=k|V_k|$
\State \Return \( G^T = (n', k, p) \)
\end{algorithmic}
\end{algorithm}

The crucial observation is that by considering subsets of agents \( V'_i \subseteq V_i \) such that \( |V'_i| = |V_k| \), we do not affect the edge distribution, 
i.e., edges of weight $-n$ and $1$ are still present with probability $p$ and $1-p$, respectively.
Our analysis for both the low and high perturbation regimes proceed by bounding how much the maximum social welfare changes because of the changes in the sizes. 

\subsubsection{Low Perturbation Regime for Random Balanced Multipartite Graphs}
\Cref{alg:balancedgreedy} takes as input a random balanced \( k \)-partite graph, $G=(\{V_1, \cdots, V_k\},p)$ where $p=\orderof(\frac{1}{k})$. It then applies \Cref{alg:balancedgreedy-phase1} to $G$, resulting in a random Tur\'{a}n graph $G^T$. In this process, \Cref{alg:balancedgreedy-phase1} selects an arbitrary subset of agents from each color class \( V'_i \subseteq V_i \) such that \( |V'_i| = |V_k| \) for all \( i \in [k] \). Let \( G^T = (n', k, p) \) denote the Tur\'{a}n graph induced by the vertices in \( \bigcup_{i=1}^{k} V'_i \), where \( n' = k |V_k| \). Note that \( G^T \) is a random Tur\'{a}n graph, as the edge weight between any pair of agents within the same color class is \(-n\) and for pairs of agents belonging to different color classes, the edge weights were drawn from a distribution that assigns a weight of \(-n\) with probability \( p \) and \( 1 \) with probability \( 1 - p \).
This distribution is ensured by not looking at any edge weights when selecting the subsets $V'_i$.

Since \( G \) is balanced, we know that \( |V_k| \geq q |V_1| \). Recall that $|V_1|\geq |V_i|$ for all $i\in[k]$, therefore \( |V_k| \geq q |V_i| \) for all \( i \in [k] \). Summing this inequality over all \( i \), we have \[ k|V_k| \geq q \sum_{i\in[k]}|V_i|. \] Given that \( \sum_{i\in[k]}|V_i| = n \) and \( k|V_k| = n' \), it follows that \( n' \geq nq \). Since \( G^T \) is a subgraph of \( G \), we also have \( n' \leq n \), which implies \( n' = \Theta(n) \).

After completing the first phase, \Cref{alg:balancedgreedy} proceeds by applying \Cref{alg:efficientgreedy} on $G^T$ to partition the agents in \( \bigcup_{i=1}^{k} V'_i \), and then assigns any remaining agents to singleton coalitions.

\begin{algorithm}
\caption{Constant-factor approximation algorithm for random balanced $k$-partite graph}\label{alg:balancedgreedy}
\begin{flushleft}
    \textbf{Input:} $\langle G,\eps\rangle$ where $G=(\{V_1, \cdots, V_k\},p)$ is a random balanced $k$-partite graph and $p=\orderof(\frac{1}{k})$\\
    \textbf{Output:} Partition $\pi$
\end{flushleft}
\begin{algorithmic}[1]
\State Reduce to a random Tur\'{a}n graph \( G^T = (n', k, p) \) by running \Cref{alg:balancedgreedy-phase1} on $\langle G,\eps\rangle$, denote $k$ color classes as $V'_i$ for $i \in [k]$.
\State Let $\pi$ be the partition on vertices of $G^T$ after running \Cref{alg:efficientgreedy} on $(G^T, S = \{V'_1,\cdots,V'_k\}, \eps)$. 
\ForAll{$v \in \bigcup_{i\in[k]} V_i \setminus \left\{ \bigcup_{i\in [k]} V'_i \right\}$}
    \State $\pi \gets \pi \cup \{v\}$
\EndFor
\State \Return $\pi$
\end{algorithmic}
\end{algorithm}

\begin{theorem}\label{thm:lowperturb}
    Consider aversion-to-enemies games given by random balanced $k$-partite graphs $G = (\{V_1,\cdots,V_k\}, p)$, where $p=\orderof(\frac{1}{k})$. 
    Then there is a polynomial-time algorithm that returns a constant-factor approximation to maximum welfare with probability $1-e^{-\Theta(n)}$.
\end{theorem}
\begin{proof}
Since \( p = \orderof(\frac{1}{k}) \), \Cref{lem:lowerbound} implies that \Cref{alg:efficientgreedy} returns \( \alpha \frac{n'}{k} \) cliques, each containing at least \( k\sqrt{1-\eps} \) agents, with probability \( 1 - e^{-\Theta(n')} \), for any constant \( \alpha \in (0, 1) \). Each clique contains at least \( k\sqrt{1-\eps} \) agents, and the utility of each agent in such a clique is at least \( k\sqrt{1-\eps} - 1 \). Therefore, the social welfare of the partition returned by the algorithm is bounded from below as follows:
\[
\SW(\pi) \geq \alpha \frac{n'}{k} k\sqrt{1-\eps}(k\sqrt{1-\eps} - 1) = \alpha n' k (1-\eps) - \alpha n' \sqrt{1-\eps}.
\]
On the other hand, \Cref{prop:n(k-1)} implies \( \SW(\pi^*) \leq n(k - 1) < nk \).
This gives the following bound on the ratio of social welfare:
\[
\frac{\SW(\pi)}{\SW(\pi^*)} \geq \frac{n'}{n} \left[ \alpha (1-\eps) - \frac{\alpha \sqrt{1-\eps}}{k} \right] \geq q \alpha (1-\eps) - \frac{q \alpha \sqrt{1-\eps}}{k},
\]
where the last inequality follows from \( n' \geq nq \). Note that for both constant and sublinear values of $k$, the approximation factor is a constant. 
\end{proof}

\subsubsection{High Perturbation Regime for Random Balanced Multipartite Graphs} 
\Cref{alg:balancedgreedy2} takes as input a random balanced \( k \)-partite graph \( G = (\{V_1, \cdots, V_k\}, p) \), where \( p = c \) for some constant \( c \in (0,1) \). It first applies \Cref{alg:balancedgreedy-phase1} to \( G \), producing a random Tur\'{a}n graph \( G^T \). Then, it applies \Cref{alg:efficientgreedy2} to partition the agents in \( G^T \), and then assigns any remaining agents to singleton coalitions.

\begin{algorithm}
\caption{Logarithmic-factor approximation algorithm for random balanced $k$-partite graph}\label{alg:balancedgreedy2}
\begin{flushleft}
    \textbf{Input:} $\langle G,\eps\rangle$ where $G=(\{V_1, \cdots, V_k\},p)$ is a random balanced $k$-partite graph and $p=c$ for a constant $c\in(0,1)$\\
    \textbf{Output:} Partition $\pi$
\end{flushleft}
\begin{algorithmic}[1]
\State Reduce to a random Tur\'{a}n graph \( G^T = (n', k, p) \) by running \Cref{alg:balancedgreedy-phase1} on $\langle G,\eps\rangle$, denote $k$ color classes as $V'_i$ for $i \in [k]$.
\State Let $\pi$ be the partition on vertices of $G^T$ after running \Cref{alg:efficientgreedy2} on $(G^T, \eps)$. 
\ForAll{$v \in \bigcup_{i\in[k]} V_i \setminus \left\{ \bigcup_{i\in [k]} V'_i \right\}$}
    \State $\pi \gets \pi \cup \{v\}$
\EndFor
\State \Return $\pi$
\end{algorithmic}
\end{algorithm}

\begin{theorem}\label{thm:highperturb}
    Consider aversion-to-enemies games given by random balanced $k$-partite graphs $G = (\{V_1\cdots,V_k\},p)$, where $p=c$ for some constant $c\in(0,1)$. 
    Then there exists a polynomial-time algorithm that returns a partition that provides a $\orderof(\log n)$-approximation of the maximum welfare with probability $1-ne^{-\Theta(\frac{n}{k})}-(\frac{ce}{2\log_e n})^{\frac{2}{c} \log_e n}$. 
\end{theorem}

\begin{proof}
    \Cref{lem:alg2bound} implies that \Cref{alg:efficientgreedy2} returns a partition \( \pi \) where \( \SW(\pi) = \Omega(n') \) with probability \( 1 - n' e^{-\Theta(\frac{n'}{k})} \). Since \( n' = \Theta(n) \), we also have \( \SW(\pi) = \Omega(n) \) with the same probability. 

 \Cref{prop:n(k-1)} implies \( \SW(\pi^*) \leq n(k-1) \). If \( k = \orderof(\log n) \), then \( \SW(\pi^*) = \orderof(n \log n) \). When \( k = \Omega(\log n) \), we add \( |V_1| - |V_i| \) dummy vertices to each color class \( V_i \) (for all \( i \in [k] \)) to create new color classes \( V^n_1, \cdots, V^n_k \), ensuring \( |V^n_i| = |V_1| \) for all \( i \in [k] \). For each pair of vertices without an edge, add a weighted edge based on the distribution defined in \Cref{def:rule}.
 Let \( G^n \) be the resulting random Tur\'{a}n graph, and let \( n'' \) represent the total number of vertices in \( G^n \). Note that \( n \leq n'' = k |V_1| \leq k \frac{|V_k|}{q} = \frac{n'}{q} \leq \frac{n}{q} \), implying \( n'' = \Theta(n) \). The maximum welfare in \( G^n \) serves as an upper bound for the maximum welfare in \( G \), since \( G \) is a subgraph of \( G^n \). When \( k = \Omega(\log n) \), this implies \( k = \Omega(\log n'') \). Therefore, by \Cref{lem: largeKSWOPT}, the maximum welfare in \( G^n \) is \( \orderof(n'' \log n'') \) with probability \( 1 - \left(\frac{ce}{2 \log_e n''}\right)^{\frac{2}{c} \log_e n''} \). Thus, in any case, \( \SW(\pi^*) = \orderof(n \log n) \).

By union bound, with probability at least 
$$ 1 - n' e^{-\Theta(\frac{n'}{k})} - \left(\frac{ce}{2 \log_e n''}\right)^{\frac{2}{c} \log_e n''} \geq 1 - n e^{-\Theta(\frac{n}{k})} - \left(\frac{ce}{2 \log_e n}\right)^{\frac{2}{c} \log_e n}$$
we obtain
\( \SW(\pi) \geq \Omega\left(\frac{1}{\log n}\right) \SW(\pi^*) \).
\end{proof}

\section{Conclusion}

We have investigated maximizing social welfare in additively separable hedonic games.
This is known to be a very hard problem in a worst-case analysis: approximating welfare better than the $n$-approximation provided by maximum weight matchings faces computational boundaries.
We have strengthened the existing approximation hardness to games with bounded valuations, in particular when restricting them to $\{-1,0,1\}$.

By contrast, we have carved out various possibilities to obtain better approximation guarantees.
In games with nonnegative total value, a randomized polynomial-time algorithm achieves a $\orderof(\log n)$-approximation.
Our proof establishes an interesting connection to the correlation clustering literature.
Moreover, going beyond worst-case guarantees, we have defined two stochastic models of aversion-to-enemies games, i.e., the games which cause the inapproximability in the first place.
In both models, we perform a high probability analysis.
The first stochastic model is based on Erd\H{o}s-R\'{e}nyi graphs, where we can efficiently compute partitions that approximate social welfare within a constant factor.
The second stochastic model is based on balanced multipartite graphs.
We distinguish between a low and high perturbation regime, where we can guarantee constant and logarithmic approximations, respectively.

Social welfare is a fundamental objective in ASHGs that deserves further attention in future research.
A specific open question is to investigate whether efficient approximation algorithms are possible for symmetric ASHGs with valuations of $\{-1,1\}$, see our discussion after \Cref{thm:hardnessapprox}.
Moreover, considering welfare approximation in suitable classes of random hedonic games might lead to intriguing discoveries.
One candidate are random ASHGs with uniformly random valuations \cite{BuKr24a}.
Finally, we restricted attention to symmetric games, which is \emph{not} without loss of generality for aversion-to-enemies games (cf.~\Cref{fn:sym}).
Hence, another direction is to consider asymmetric subclasses of ASHGs.

\section*{Acknowledgements}
    Martin Bullinger is supported by the AI Programme of The Alan Turing Institute, Vaggos Chatziafratis is supported by a UCSC startup grant and a Hellman's fellowship, and Parnian Shahkar is supported by the National Science Foundation (NSF) under grant CCF-2230414.

\newcommand{\etalchar}[1]{$^{#1}$}

\end{document}